%% file: Zhong-etal_ExactMaxID.tex
\documentclass[12pt]{article}   

\RequirePackage[OT1]{fontenc}
\usepackage{natbib}
\RequirePackage[colorlinks,citecolor=blue,urlcolor=blue]{hyperref}
\usepackage[english]{babel}
\usepackage{amsthm}
\usepackage{amsmath}
\usepackage{amsfonts}
\usepackage{amssymb}
\usepackage{graphicx}
\usepackage{enumerate}
\usepackage{color}
\usepackage{array}
\usepackage{tikz}
\usepackage{comment}
\usepackage{bm}
 \usepackage{algorithm}
\usepackage{algpseudocode}
\algrenewcommand\algorithmicrequire{\textbf{Input:}}
\algrenewcommand\algorithmicensure{\textbf{Output:}}

\definecolor{myorange}{HTML}{E24100}
\definecolor{mygreen}{HTML}{228b22}

\usepackage{ulem} 

\input{macros.tex}

\input{teachingmacros.tex}

\def\Rev#1{{#1}}
\def\Rem#1{{\iffalse #1 \fi}}
\def\Rep#1#2{{#2}}

\begin{document}
\thispagestyle{empty}
\baselineskip=28pt
\vskip 5mm
\begin{center} 
{\Large{\bf Exact Simulation of Max-Infinitely Divisible Processes}}
\end{center}

\baselineskip=12pt 
\vskip 5mm

\begin{center}
\large
Peng Zhong$^1$, Rapha\"el Huser$^1$, Thomas Opitz$^2$
\end{center}

\footnotetext[1]{
\baselineskip=10pt Computer, Electrical and Mathematical Sciences and Engineering (CEMSE) Division, King Abdullah University of Science and Technology (KAUST), Thuwal 23955-6900, Saudi Arabia. E-mails: peng.zhong@kaust.edu.sa; raphael.huser@kaust.edu.sa}
\footnotetext[2]{
\baselineskip=10pt BioSP, INRAE, Avignon, 84914, France, E-mail: thomas.opitz@inrae.fr}

\baselineskip=17pt
\vskip 4mm
\centerline{\today}
\vskip 6mm

\begin{center}
{\large{\bf Abstract}}
\end{center}
Max-infinitely divisible (max-id) processes play a central role in extreme-value theory and include the subclass of all max-stable processes. They allow for a constructive representation based on the pointwise maximum of random functions drawn from a Poisson point process defined on a suitable {function} space. Simulating from a max-id process is often difficult due to its complex stochastic structure, while calculating its joint density in high dimensions is often numerically infeasible. Therefore, exact and efficient simulation techniques for max-id processes are useful tools for studying the characteristics of the process and for drawing statistical inferences. Inspired by the simulation algorithms for max-stable processes, \Rep{we here develop theory and algorithms to generalize simulation approaches tailored for certain flexible (existing or new) classes of max-id processes.}{theory and algorithms to generalize simulation approaches tailored for certain flexible (existing or new) classes of max-id processes are presented.} Efficient simulation for a large class of models can be achieved by implementing an adaptive rejection sampling scheme to sidestep a numerical integration step in the algorithm. \Rep{We present the}{The} results of a simulation study \Rep{highlighting}{highlight} that our simulation algorithm works as expected and is highly accurate and efficient, such that it clearly outperforms customary approximate sampling schemes. \Rep{As a byproduct we also develop here new max-id models, which can be represented as pointwise maxima of general location-scale mixtures, and which possess flexible tail dependence structures capturing a wide range of asymptotic dependence scenarios.}{As a by-product, new max-id models, which can be represented as pointwise maxima of general location-scale mixtures and possess flexible tail dependence structures capturing a wide range of asymptotic dependence scenarios, are also developed.}
\baselineskip=16pt

\par\vfill\noindent

\pagenumbering{arabic}
\baselineskip=24pt

\newpage


\allowdisplaybreaks

\section{Introduction}
\label{sec:introduction}
Max-infinitely divisible processes (max-id) have gained much popularity in recent years for the flexible modeling of \Rev{spatiotemporal} extremes in phenomena such as precipitation, wind speeds and temperatures \citep{Padoan2013,Huser.etal:2020,Bopp.etal:2021, Huser.Wadsworth:2020, Zhong.Huser.Optiz:2022}. Such models encompass max-stable processes \citep{Haan1984} as a subclass, which are currently the mainstream for modeling spatial \Rev{and temporal} extremes, as demonstrated by a large body of theoretical and applied literature \Rev{\citep[see, e.g.,][]{Davison.etal:2012,Opitz:2013,Huser.Davison:2014,Dey2016,Gissibl.et.al:2018,Davison.etal:2019}}. Max-stability is a property that arises in asymptotic models for multivariate block maxima, but that is however often violated in finite samples \citep{Huser.Wadsworth:2020}. Precisely, the dependence strength in max-stable processes does not depend on event magnitude, while most environmental and climatic processes suggest that it should become weaker at higher thresholds, with the most extreme events being typically more localized. To circumvent this limitation, more flexible max-id processes, which drop the restrictive max-stability assumption but retain natural properties of multivariate block maxima, have been considered. In particular, the max-id models proposed by \citet{Huser.etal:2020}, \citet{Bopp.etal:2021} and \citet{Zhong.Huser.Optiz:2022}, extend certain popular classes of max-stable models (obtained as limiting cases on the boundary of the parameter space), in order to achieve a good compromise between the pragmatism of flexible max-id models and the strong theoretical foundations of the max-stable sub-class. 

The theory behind max-id processes has been studied in depth \citep{Brown.Resnick:1977,Resnick1987,Gine1990}. In particular, \cite{Gine1990} showed that max-id processes possess a functional Poisson point process (PPP) representation. More precisely, each max-id process $Z = \{Z(\bm s)\}_{\bm s \in \calS}$ defined over the region $\calS\subset\mathbb{R}^d$, can be defined as the pointwise maximum of a potentially infinite number of random functions $\{\eta_i; \ i=1,2,\ldots\}$ defined over $\calS$, which are sampled according to a PPP with mean measure {$\Lambda$}, i.e.,
\begin{equation}
\label{eq:max-id}
Z(\bm s)=\max_{i=1,2,\ldots}\eta_i(\bm s),\quad \bm s\in\mathcal{S},    
\end{equation}
{where, by convention, $Z(\bm s)$ takes value at the lower boundary of the support of $\eta_i(\bm s)$ when the PPP $\{\eta_i\}$ contains no point (which can happen when $\Lambda$ is finite). To avoid intricate theoretical issues dealing with the lower boundary, we here assume that $\Lambda$ is an infinite measure, such that the number of Poisson points $\eta_i$ in \eqref{eq:max-id} is infinite.} 
However, analytical forms of the multivariate {density} functions of such processes exist only in special cases, and by analogy with max-stable processes their numerical evaluation in high dimensions becomes impossible due to the combinatorial explosion of the number of terms to be calculated \citep{Padoan.etal:2010,Castruccio.etal:2016,Huser.etal:2020}. Consequently, simulation of max-stable processes is often crucial to study their dependence characteristics, or for simulation-based inference \citep{Erhardt2012,Hainy.etal:2016,Lee.etal:2018}, and the same is true for general classes of max-id processes. However, exact simulation of $Z(\bm s)$ in \eqref{eq:max-id} is not trivial, especially when the number of functions $\eta_i$ is infinite and when these random functions are tricky to sample.

For the subclass of max-stable models, there are essentially two types of generic simulation algorithm that have been proposed based on variants of model representations stemming from \eqref{eq:max-id} \citep{Oesting.Strokorb:2021}.  These algorithms exploit the specific structure of max-stable processes---the so-called \emph{spectral representation} \citep{Haan1984}, whereby on the unit Fr\'echet scale (i.e., {$\pr\{Z(\bm s)\leq z\}=\exp(-1/z)$, $z>0$}) the points of the Poisson process $\{\eta_i;i=1,2,\dots\}$ in \eqref{eq:max-id} can be decomposed into two multiplicative and stochastically independent components as $\eta_i(\bm s)=R_iW_i(\bm s)$. Precisely, $\{R_i;\ i=1,2,\dots\}$ here denotes a Poisson point process on $[0,\infty)$ with {mean measure} $r^{-2}{\rm d}r$, independent of the random functions $\{W_i;\ i=1,2,\dots\}$, which are independent copies of a random process $W(\bm s)$ defined on $\calS$ such that $\E[\max\{0,W(\bm s)\}]=1$. This decomposition always exists for max-stable processes with unit Fr\'echet margins but is not unique. When the points $\{R_i;\ i=1,2,\dots\}$ have heavy-tailed intensity ${\alpha} r^{-(\alpha+1)}{\rm d}r$ for some $\alpha>0$, the resulting process $Z$ remains max-stable, but with $\alpha$-Fr\'echet margins. Moreover, the points $\{R_i^{-1};\ i=1,2,\dots\}$ can be defined as the arrival times of a renewal process such that we can simulate them in decreasing order, i.e., $R_1>R_2>\ldots$. When this is the case, the componentwise maximum of $\{R_jW_j; \ j=1,\ldots,i-1\}$ is less and less impacted by the $i$-th process $\eta_i=R_iW_i$ as $i>1$ is iteratively incremented. If $W(\bm s)$ is an almost surely bounded random process on $\calS$, the contribution of $\eta_i=R_iW_i$ becomes even ultimately completely irrelevant{, as $i$ grows}. Based on this idea, \cite{Schlather:2002} introduced the first exact simulation algorithm for max-stable processes by setting an appropriate stopping rule for the simulation of the Poisson points $R_iW_i$. 
If $W(\bm s)$ is unbounded, we can use certain alternative representations  to make it a bounded random process under a suitable change of measure \citep{Dieker.Mikosch:2015,Dombry.Engelke.Oesting:2016,Oesting.Schlather.Zhou:2018}, but simulation from the bounded process is intricate and numerically challenging  even in moderate dimensions for many of the customary models \citep{Kabluchko.etal:2009,Opitz:2013}. The second type of exact simulation algorithm for max-stable processes, introduced by \cite{Dombry.Engelke.Oesting:2016}, relies on the concept of extremal functions \citep{Dombry2013}. With this method, it is possible to directly simulate the random functions $\{\eta_i\}$ on $\calS$ that contribute to the maximum at a finite number of locations. In practice, efficient conditional simulation of the random functions $\eta_i(\bm s)$ given its value at a single location $\bm s_0$ is a prerequisite for efficient simulation with this approach.

However, none of the approaches described above tackled exact simulation for general max-id processes, and we intend to fill this gap. In this work, we demonstrate that the idea of the exact simulation algorithm based on extremal functions, developed in \cite{Dombry.Engelke.Oesting:2016} for max-stable processes carries over to the general max-id case after suitable adjustments, and we illustrate it with max-id models that have found interest in practical applications. The paper is organized as follows. In \S\ref{sec:Max-id}, we recall some theoretical background and formally define the notions of max-id processes and extremal functions. Then, in \S\ref{sec:Exact-Simu}, we derive the specificities of the theory of max-id processes required for exact simulation based on extremal functions, and we present simulation algorithms focusing on two major classes of max-id processes. In particular, we start with some known max-id models of the form \eqref{eq:max-id} where the random functions $\eta_i(\bm s)$ can be represented as Gaussian scale mixtures, and we then also explore new max-id models based on Gaussian location mixtures that possess appealing tail dependence characteristics. In \S\ref{sec:Simu-Study}, we conduct a simulation study to demonstrate the performance of our proposed simulation algorithm. We conclude with a discussion in \S\ref{sec:Discussion}.

\section{Max-id Processes and Extremal Functions}
\label{sec:Max-id}
A process $\{Z(\bm s)\}_{\bm s \in \calS}$ is called max-infinitely divisible (max-id) if the joint cumulative distribution function (cdf) $G$ of $\bm Z(K_N)=\{Z(\bm s_1),\ldots, Z(\bm s_N)\}^\top$ at any finite collection of sites $K_N=\{\bm s_1,\bm s_2,\dots,\bm s_N\}\subset \calS$ defines a valid cdf $G^t$ (with the notation $G^t(\bm z)=\{G(\bm z)\}^t$, $\bm z=(z_1,\ldots,z_N)^\top\in \mathbb R^N$) for any  $t>0$. The distribution $G^t$ does not necessarily stay within the same location-scale family;\ this property is only satisfied for max-stable distributions $\tilde G$, for which $\tilde G^m(\bm a_m\bm z+\bm b_m)=\tilde G(\bm z)$, for any integer $m=1,2,\dots$, and appropriately defined normalizing vectors $\bm a_m\in (0,\infty)^{N}$ and $\bm b_m\in\mathbb R^N$. As shown in Equation~\eqref{eq:max-id}, any max-id process can be constructed by taking pointwise maxima over a Poisson point process defined on a suitable function space. We assume that $\{Z(\bm s)\}_{\bm s \in \calS}$ is a sample-continuous max-id process in a function space $\calC$, where $\calC$ denotes the space of continuous functions with compact support $\calS \subset \mathbb{R}^d$ endowed with the uniform norm, i.e., \Rev{$\|f\|=\sup_{\bm s\in \calS} |f(\bm s)|, f\in \calC$}, and $\{\eta_i(\bm s);\  i=1,2,\dots\}_{\bm s\in\calS}$ in \eqref{eq:max-id} is a Poisson point process with mean measure $\Lambda$ on $\calC$. Then, the extremal functions are defined as follows. 
\begin{defn}
	 Let $K\subset \calS$ be a nonempty subset of the domain $\calS$, and $Z$ a max-id process. A function $\phi\in\calC$ is called $K$-extremal if there exists $\bm s\in K$ such that $Z(\bm s)=\phi(\bm s)$, otherwise $\phi$ is called $K$-subextremal. We denote $\calC_K^+(Z)$ the set of K-extremal functions and $\calC_K^-(Z)$ the set of $K$-subextremal functions.
\end{defn}
In other words, an extremal function is a Poisson point $\eta_i$ in $\calC$ that contributes to the maximum process $Z$ at one or more locations. Notice that both $\calC_K^+(Z)$ and $\calC_K^-(Z)$ depend on the max-id process $Z$, even though the realization of the max-id process $Z$ on $K$ is fully determined by $\calC_K^+(Z)$ alone. Let $\ell$ be the vertex function (i.e., the lower boundary) of the max-id process $Z$, defined by 
\begin{equation}\label{eq:vertex}
\ell(\bm s) = \sup\{z\in \mathbb R: \pr(Z(\bm s)\geq z) = 1\} \in [-\infty,\infty), \quad\bm s \in \calS.
\end{equation} 
\cite{Gine1990} showed that for sample continuous max-id processes in $\calC$, if the vertex function $\ell$ is continuous, then it can be subtracted from $Z$. Therefore, we here assume $\ell(\bm s)=0$ without loss of generality, and denote $\calC_0=\{f\in\calC; f\neq 0, f\geq 0\}$. We further assume that for any fixed $\bm s_0\in \calS$, 
\begin{gather}\label{eq:assumption}
	\Lambda_{\bm s_0}(z)=\Lambda(\{\eta\in\calC_0; \eta(\bm s_0)\in [z,\infty) \})\ \text{is continuous on}\ (0,\infty),\ \lim_{z\downarrow 0}\Lambda_{\bm s_0}(z)=\infty, \\
\Lambda_{\bm s_0}(\{\eta\in\calC_0; \eta(\bm s_0) > \varepsilon \}) < \infty, \forall\ \varepsilon > 0.
\end{gather}
Under these assumptions, \cite{Dombry2013} showed that for singletons $K=\{\bm s_0\}$ the set of extremal functions $\calC_{\{\bm s_0\}}^+$ contains almost surely a single point {(a function)}, denoted by $\phi_{\bm s_0}^+$, i.e., the resulting max-id ``process" $Z(\bm s_0)$ is realized by one and only one Poisson point $\eta_i$ at a given location $\bm s_0$. {This assumption implies that $Z$ has continuous margins and the marginal distribution has no mass at the lower boundary $\ell(\bm s) =0$.}

Under the standing assumptions, \citet[][Lemma A.2]{Dombry2013} show that  the mean measure $\Lambda$ possesses a \emph{regular conditional probability measure} $P_{\bm s_0}(z,\cdot)$ with respect to the location $\bm s_0$; i.e., given that $\eta_i(\bm s_0)=z>0$ for a point $\eta_i$ of the Poisson point process, the random function $\eta_i(\bm s)$ for $\bm s\not=\bm s_0$ has (conditional) probability distribution given by $P_{\bm s_0}(z,\cdot)$. \citet{Dombry2013} further show in their Theorem~3.2 that $P_{\bm s_0}(z,\cdot)$ is also the conditional distribution for the extremal function $\phi_{\bm s_0}^+$ at $\bm s_0$ given $Z(\bm s_0)=z > 0$; moreover, conditionally on the values $Z(\bm s_1),\ldots,Z(\bm s_n)$ of the max-id process at $K=\{\bm s_1,\ldots,\bm s_n\}$, the subextremal functions in $\calC_K^-(Z)$ define a Poisson point process that is independent of the extremal functions in $\calC_K^+(Z)$. This restricted Poisson point process has again mean measure $\Lambda$ but now with all the mass ``removed'' when one of the conditioning values $Z(\bm s_1),\ldots,Z(\bm s_n)$ is exceeded, i.e., it has no mass on $\{\eta \in\calC_0: \max_{i=1}^n \eta(\bm s_i)/Z(\bm s_i)>1\}$. This remarkable result is the key for the proposed simulation algorithm that consists in iteratively simulating extremal functions at a set of locations (see \S\ref{sec:generalalgo}), and we summarize it in the following lemma.
  \begin{lemma}\label{lemma}\citep[][Theorem 3.2]{Dombry2013}
  	Given the extremal functions $\calC_K^+(Z)$ on $K$, the subextremal functions $\calC_K^-(Z)$ form a Poisson point process on $\calC_0$ with intensity $\mathbb{I}(f(K)< Z(K))\Lambda(\mathrm{d}f)$, where $\mathbb{I}(\cdot)$ is the indicator function.
  \end{lemma}
  In particular, when the process $\{Z(\bm s)\}_{\bm s\in\calS}$ is max-stable with unit Fr\'echet margins, the Poisson point process $\{\eta_i;\ i=1,2,\dots\}$ can be decomposed as $\eta_i(\bm s) = R_iW_i(\bm s)$, $i=1,2,\dots$, where $\{R_i;\ i=1,2,\dots\}$ is a Poisson point process on $(0,\infty)$ with mean measure $\kappa([r,\infty])=r^{-1}, r>0$, and $\{W_i(\bm s)\}_{\bm s\in \calS},\  i=1,2,\dots,$ are independent copies of a random process $\{W(\bm s)\}_{\bm s\in \calS}$ with $\E[\max\{W(\bm s),0\}]=1$, which are also independent of the points $\{R_i;i=1,2,\dots\}$ \citep{Haan1984,Schlather:2002}. A number of parametric max-stable models have been proposed for statistical applications, such as the Brown--Resnick model \citep{Brown.Resnick:1977, Kabluchko.etal:2009}, and the extremal-$t$ model \citep{Opitz:2013}, where the intensity function associated to the mean measure $\Lambda$ has a closed form expression. In these special cases, the multivariate conditional distributions required for exact simulation are relatively easy to simulate from, 
and \cite{Dombry.Engelke.Oesting:2016} provide the exact forms of the distributions $P_{\bm s}(z,\cdot)$ required for simulating these max-stable models. 

However, such simplifications are no longer available for general max-stable processes and for the classes of max-id (but not max-stable) process models proposed in the recent literature. Very often, the mean measure $\Lambda(\cdot)$ does not have a closed-form intensity function, and the corresponding conditional distributions are complicated; rather, calculations of relevant measures of the form $\Lambda(B)$ and intensities often involve computing a one-dimensional integral \citep{Huser.etal:2020, Zhong.Huser.Optiz:2022}. Therefore, simulating from $P_{\bm s}(z,\cdot)$ is more involved than in the max-stable case and requires further investigation into its form to propose efficient simulation techniques.

\section{Exact Simulation}
\label{sec:Exact-Simu}

\subsection{General simulation algorithm}
\label{sec:generalalgo}

We now extend the generic max-stable simulation algorithm based on extremal functions, as proposed  by \citep{Dombry.Engelke.Oesting:2016}, to the general max-id case. {While our proposed algorithm in the max-id case is almost identical in its overall idea and general structure to \citeauthor{Dombry.Engelke.Oesting:2016}'s algorithm in the max-stable case (modulo certain technicalities and specificities in how the different steps are executed), it is still important to note that generalizing a specific simulation algorithm that works for max-stable processes to the much wider class of max-id models, while remaining an exact and feasible simulation procedure, is not necessarily trivial and requires careful investigation. This has never been explored before.} 

{Our aim is, thus,} to simulate the max-id process $\{Z(\bm s)\}_{s\in\calS}$ at a the finite number $N$ of sites given as $K_N=\{\bm s_1,\bm s_2,\dots,\bm s_N\}\subset \calS$. We define the subsets $K_n=\{\bm s_1,\dots,\bm s_n\}$ of the first $n$ sites for  $n=1,\dots,N$, and we denote by $Z_n = \max_{\eta_i\in \calC^+_{K_n}}\eta_i$ the process constructed from the extremal functions at the sites in $K_n$, such that  {$\bm Z(K_N)=\bm Z_N(K_N):=\{Z_N(\bm s_1),\ldots,Z_N(\bm s_N)\}^\top$}. 
To simulate from the joint distribution of the process $\{Z(\bm s)\}$ at sites in $K_N$, we can iterate through the set of locations from $\bm s_1$ to $\bm s_N$ in order to simulate the extremal functions at each location.  
{More precisely, suppose we have already simulated the extremal functions on $K_n$. Recall that we denote by $\{\phi_{\bm s_n}^+\}_{1\leq n\leq N}$	the sequence of extremal functions associated to the singletons $\{\bm s_n\}$, $n=1,\ldots,N$, where some of these functions may be identical to each other. The extremal function $\phi_{\bm s_{n+1}}^+$ at a new location $\bm s_{n+1}\notin K_n$ either coincides with one of the extremal functions on $K_n$ already simulated, i.e., $\phi_{\bm s_{n+1}}^+\in\calC_{K_n}^+(Z)$, or it corresponds to a subextremal function on $K_n$, i.e., $\phi_{\bm s_{n+1}}^+\in\calC_{K_n}^-(Z)$, in which case it is a point of the restricted Poisson point process defined in Lemma~\ref{lemma}. Moreover, at $\bm s_{n+1}$ the maximum of the values of the extremal functions on $K_n$, i.e., $Z_n(\bm s_{n+1})$, already provides a lower bound for the values of the new extremal function that could arise at $\bm s_{n+1}$. Therefore, we proceed as follows to find the extremal function $\phi_{\bm s_{n+1}}^+$, which directly provides $Z(\bm s_{n+1})$. We simulate from the unrestricted point process by assuming that its points are  ordered in descending order of the values at $\bm s_{n+1}$, i.e., $\eta_1(\bm s_{n+1})>\eta_2(\bm s_{n+1})>\ldots$. Simulation starts from the point $\eta_1$ with the largest value at $\bm s_{n+1}$ and then continues in descending order. To keep only the points from the restricted point process of the subextremal functions of $K_n$, we reject the points violating the upper bound conditions at $\bm s_{1},\ldots,\bm s_n$ set by the restrictions given in Lemma~\ref{lemma}.  We simulate  a point $\eta_i$ by first simulating its value $\eta_i(\bm s_{n+1})$ at $\bm s_{n+1}$ according to  the marginal (unrestricted) Poisson point process with mean measure $\Lambda_{\bm s_{n+1}}$, 
	and we then simulate the values at all the 
	other sites in $K_N$ using the regular conditional probability distribution  $P_{\bm s_{n+1}}(\eta_i(\bm s_{n+1}),\cdot)$. Moreover, we can stop simulating as soon as a value $\eta_i(\bm s_{n+1})$ falls below the lower bound $Z_n(\bm s_{n+1})$  and we reject it, since this value (as well as all the following smaller ones) cannot correspond to the extremal function at $\bm s_{n+1}$. In any case, we stop as soon as we obtain a point $\eta_i$ that is not rejected, and this point then corresponds to $\phi_{\bm s_{n+1}}^+$. If we have rejected all points until we stop,  there is no new extremal function at $\bm s_{n+1}$,  such that $\phi_{\bm s_{n+1}}^+\in\calC_{K_n}^+(Z)$.}
 
{
 The benefit of this approach is that it requires to simulate only a finite number of Poisson points to either obtain a new extremal function at $\bm s_{n+1}$, or to confirm that the extremal function at $\bm s_{n+1}$ coincides with one of the already simulated extremal functions on $K_n$. Therefore, the approach exploits the fact that the already simulated extremal functions provide upper bounds at $\bm s_1,\ldots,\bm s_n$ and a lower bound at $\bm s_{n+1}$.  This idea then allows us to iteratively simulate all of the extremal functions at a finite set of locations $\bm s_1,\bm s_2,\ldots,\bm s_N$, where we condition on the extremal functions at $\bm s_1,\ldots,\bm s_{n}$ when identifying the extremal function at $\bm s_{n+1}$, $1\leq n<N$.  Since the extremal functions define the values of the max-id process at the locations $\bm s_1,\bm s_2,\ldots,\bm s_N$, this is all we need. \citet{Dombry.Engelke.Oesting:2016} have proposed this approach for  max-stable processes, but their algorithm generalizes straightforwardly to max-id processes thanks to the general results of \citet{Dombry2013}.} 
The following theorem {formalizes the simulation procedure detailed above by characterizing} the conditional distributions of $\phi_{\bm s_{n}}$ given $\{\phi_{\bm s_k}^+\}_{1\leq k\leq n-1}$  for extremal functions at position $n=2,\ldots,N-1$, which are required for iterative simulation. 

\begin{theorem}[Conditional distributions of extremal functions]\label{thm:2}
Consider a max-id process  $\{Z(\bm s)\}_{s\in\calS}$  as defined in \S\ref{sec:Max-id}. Given sites $K_N=\{\bm s_1,\bm s_2,\dots,\bm s_N\}$,  the distribution of its extremal function $\phi_{\bm s_{1}}^+$ and of its conditional extremal functions {$\phi_{\bm s_{n+1}}^+\mid\{\phi_{\bm s_k}^+\}_{1\leq k\leq n}$ for $n\geq 1$} can be characterized as follows:
\begin{enumerate}[(1)]
    \item Initial extremal function: given a realization $z$ of $Z(\bm s_1)$, the extremal function $\phi_{\bm s_1}^+$ is distributed according to $P_{\bm s_1}(z,\cdot)$.
    \item Conditional extremal functions for $n\geq 1$: given the extremal functions $\{\phi_{\bm s_i}^+\}_{1\leq i \leq n}$ for sites in $K_{n}$, consider the Poisson measure  of functions  that  do not exceed the already simulated max-id values of sites in $K_{n}$ but that exceed the maximum of already simulated values at $\bm s_{n+1}$:
     \begin{equation}\label{eq:simu}
      \tilde\Lambda(\mathrm{d}f)=\mathbb{I}\left(f(\bm s_i)<Z_{n}(\bm s_i),1\leq i\leq n\right)\mathbb{I}\left(f(\bm s_{n+1})>Z_n(\bm s_{n+1})\right)\,\Lambda(\mathrm{d}f).
      \end{equation}
      Then, {we denote the Poisson point process with mean measure $\tilde{\Lambda}$ by ${\rm PPP}(\tilde\Lambda)$,} and the extremal function $\phi_{\bm s_{n+1}}^+$ conditional on $\{\phi_{\bm s_i}^+\}_{1\leq i \leq n}$ is given as follows: 
    \[\tilde \phi_{\bm s_{n+1}}^+ = \left\{\begin{array}{ll}
    \arg\max_{\phi\in {\rm PPP}(\tilde\Lambda)}\phi(\bm s_{n+1}), & {{\rm PPP}(\tilde\Lambda)\not=\emptyset,}\\
     \arg\max_{\phi\in\{\phi_{\bm s_1}^+,\dots,\phi_{\bm s_{n}}^+\}}\phi(\bm s_{n+1}),    & 
     {{\rm PPP}(\tilde\Lambda)=\emptyset\Rep{,}{.}}
    \end{array}\right.\]
\end{enumerate}
\end{theorem}

\begin{proof}
{The proof essentially follows the intuitive description of the algorithm above. It is almost identical (modulo some technicalities) to the proof of Theorem 2 in \cite{Dombry.Engelke.Oesting:2016}, which concerns max-stable processes, while we here treat max-id processes.} 

As mentioned, we proceed sequentially to simulate extremal functions. The distribution of the first extremal function $\phi_{\bm s_1}^+$ given $\phi_{\bm s_1}^+(\bm s_1) = z$ is exactly $P_{\bm s_1}(z,\cdot)$, as explained in \S\ref{sec:Max-id} above, and $\phi_{\bm s_1}^+(\bm s_1)$ thus follows the marginal distribution of $Z(\bm s_1)$, which has now a more general form than for max-stable processes. Notice also that given the continuity assumption in \eqref{eq:assumption}, there is almost surely exactly one extremal function at each location.

{Then, to simulate the extremal function $\phi_{\bm s_{n+1}}^+$ at site $\bm s_{n+1}$, for $n\geq 1$, we recall that $\phi_{\bm s_{n+1}}^+$ is either a subextremal function in $\calC_{K_n}^-$ or an extremal function in $\calC_{K_n}^+$. According to Lemma~\ref{lemma}, given $\calC_{K_n}^+$, the conditional distribution of $\calC_{K_n}^-$ follows the distribution of a Poisson point process with intensity  
 \begin{equation}\label{eq:restricted-ppp}
\mathbb{I}\left(f(\bm s_i)<Z_{n}(\bm s_i),1\leq i\leq n\right)\,\Lambda(\mathrm{d}f).
\end{equation}
{If $\phi_{\bm s_{n+1}}^+\in\calC_{K_n}^-$, we need $\phi_{\bm s_{n+1}}^+(\bm s_{n+1})>Z_n(\bm s_{n+1})$, such that $\phi_{\bm s_{n+1}}^+$ is the extremal function at $\bm s_{n+1}$.  Therefore, we can further restrict the Poisson process with intensity in \eqref{eq:restricted-ppp} and  consider the Poisson process ${\rm PPP}(\tilde\Lambda)=\calC_{K_n}^-\cap \{f\in\calC_0: f(\bm s_{n+1}) > Z_n(\bm s_{n+1})\}$, which has the mean measure $\tilde{\Lambda}$ given in Equation~\eqref{eq:simu}. If there is no point in ${\rm PPP}(\tilde\Lambda)$, i.e. $\phi_{\bm s_{n+1}}^+\notin\calC_{K_n}^-\cap \{f\in\calC_0: f(\bm s_{n+1}) > Z_n(\bm s_{n+1})\}$, it must be true that $\phi_{\bm s_{n+1}}^+\in\calC_{K_n}^+$ and  $\phi_{\bm s_{n+1}}^+=\arg\max_{\phi\in\{\phi_{\bm s_1}^+,\dots,\phi_{\bm s_{n}}^+\}}\phi(\bm s_{n+1})$.  Otherwise, $\phi_{\bm s_{n+1}}^+(\bm s_{n+1})$ is a point of ${\rm PPP}(\tilde\Lambda)\not=\emptyset$, and it can be identified according to the procedure outlined above: we simulate from the unrestricted Poisson process and reject the points violating the restriction  in Equation~\eqref{eq:simu} until either the restriction is satisfied or the simulated value falls below the lower bound $Z_n(\bm s_{n+1})$.}}
\end{proof}

{The pseudo code of an} exact simulation algorithm based on Theorem~\ref{thm:2} is given in Algorithm~\ref{algo}. 
Although the main structure of the algorithm resembles that proposed by \citet{Dombry.Engelke.Oesting:2016} for max-stable processes, the key difficulty for general max-id processes resides in efficiently simulating from the distribution $P_{\bm s}(z,\cdot)$ given $Z(\bm s)=z$. For suitable constructions of the {PPP} $\{\eta_i;\, i=1,2,\dots\}$, we can investigate the specific structure of this conditional distribution. In this paper, we consider two different and very general structures for $\eta_i$, which are inspired from the construction of max-stable processes, yet they give rise to max-id models with much more flexibility in their joint tail decay rates. These two types of models, which encompass all of the commonly used max-stable models from the spatial extremes literature, are described in the following two subsections. 

\begin{algorithm}[h]
  \caption{Exact simulation of a max-id process $Z$ at locations $K_N=\{\bm s_1,\dots,\bm s_N\}.$\label{algo}}
  \begin{algorithmic}[1]
  \Require  Stationary marginal distribution $G_0$ of $Z$. 
  \Require  Dependence characterization of $Z$  allowing simulation of $P_{\bm s_n}$, $n=1,\ldots,N$. 
  \State Simulate $E_0 \sim \rm{Exp}(1)$ and set $z=G_0^{-1}(\exp(-E_0))$.
  \State Simulate $Y\sim P_{\bm s_1}(z,\cdot)$ {over $K_N$, which yields $\bm Y(K_N)=\{Y(\bm s_1),\ldots,Y(\bm s_N)\}^\top$}. 
  \State Set $\bm Z(K_N) = \bm Y(K_N)$.
  \ForAll {$n=2,\dots,N$}
  \State Simulate $E_0 \sim \rm{Exp}(1)$ and set $z=G_0^{-1}(\exp(-E_0))$.
  \While {$\{z>Z(\bm s_n)\}$}
  \State Simulate $Y\sim P_{\bm s_n}(z,\cdot)$ {over $K_N$, which yields $\bm Y(K_N)=\{Y(\bm s_1),\ldots,Y(\bm s_N)\}^\top$}.
  \If {$Y(\bm s_i)<Z(\bm s_i)$ for all $i=1,\dots,n-1$}
  \State update simulated values as $\bm Z(K_N)=\max(\bm Z(K_N),\bm Y(K_N))$ {(componentwise)}.
  \EndIf
  \State Simulate $E_1 \sim \rm{Exp}(1)$ and update $E_0$ and $z$ by setting $E_0 = E_0 + E_1$ and $z=G_0^{-1}(\exp(-E_0))$.
  \EndWhile
  \EndFor
  \Ensure Return $\bm Z(K_N)$
  \end{algorithmic}
 \end{algorithm}

\subsection{Representation of $\eta_i$ as a Gaussian scale mixture}
\label{sec:gaussian-scale-mixture}

\subsubsection{Model description} \label{sec:gaussian-scale-mixture-description}
One possible structure of the general max-id process is achieved by representing $\eta_i$ as a \emph{Gaussian scale mixture}, i.e., $\eta_i = R_iW_i$, where $\{R_i\}$ are the points of a Poisson point process on $\mathbb R^+$ with {mean measure $\kappa$ such that $\kappa([0,\infty))=\infty$}, and $\{W_i\}$ are independent {standard} Gaussian random fields {with mean zero and unit variance} on $\calS$, which are not necessarily identically distributed and may depend on the overall random ``magnitude'' $R_i$. This construction mimics  the max-stable spectral construction of extremal-$t$ processes \citep{Opitz:2013} arising for a specific choice of $\Lambda$ {in \eqref{eq:max-id}} with $\{R_i\}$ and $\{W_i\}$ being independent. To extend max-stable processes, \citet{Huser.etal:2020} proposed using a more flexible mean measure for $\{R_i\}$ and \citet{Zhong.Huser.Optiz:2022} further relaxed the independence assumption between $\{R_i\}$ and $\{W_i\}$. 

Specifically, they proposed a Weibull-tailed mean measure $\kappa$ for $\{R_i\}$, where
\begin{equation}\label{eq:measure_R}
\kappa(r):=\kappa\left([r,\infty)\right) = r^{-\beta}\exp\{-\alpha(r^{\beta}-1)/\beta\},\quad r>0,\quad (\alpha,\beta)^\top\in (0,\infty)^2. 
\end{equation}
With this specification, $\kappa(r)\to r^{-\alpha}$ as $\beta\downarrow0$. Therefore, the model remains in the ``neighborhood'' of max-stable processes for which $\kappa(r)=r^{-\alpha}$. {Recall that the extremal-$t$ max-stable process with $\alpha$-Fr\'echet margins can indeed be obtained by taking $\eta_i=R_iW_i$ in \eqref{eq:max-id} with $\{R_i\}$ a Poisson point process on $(0,\infty)$ with mean measure $\kappa(r)=r^{-\alpha}$, and $\{W_i\}$, independent zero mean Gaussian processes, that are rescaled such that $\E[\max\{0,W_i(\bm s)\}]=1$, and that are also independent of $\{R_i\}$. In other words, if $\{R_i\}$ and $\{W_i\}$ are independent of each other and $W_i$ is Gaussian, then the max-id model constructed from \eqref{eq:measure_R} can be arbitrarily close to the max-stable extremal-$t$ dependence structure by varying the parameter $\beta$.} Given $\{R_i\}$, {a natural choice is, thus, to specify} the processes $\{W_i\}$ {to be} standard Gaussian random fields with {some} correlation function $\rho(\bm s_1,\bm s_2;R_i)$, where $R_i$ modulates the dependence range of $W_i$ in such a way that large magnitudes $R_i$ create processes $W_i$ with weaker dependence. This gives extra flexibility for capturing the joint tail decay rate of the resulting max-id process $Z$ in the asymptotic independence setting, i.e., when $\lim_{u\to1}\Pr[G_0\{Z(\bm s_1)\}>u\mid G_0\{Z(\bm s_2)\}>u]=0$ with $G_0$ being the marginal distribution of $Z$. For example, \citet{Zhong.Huser.Optiz:2022} chose a non-stationary correlation $\rho(\bm s_1,\bm s_2;R_i)$, which may be specified as
\begin{equation}\label{eq:corr_function_W}
\rho(\bm s_1,\bm s_2;R_i) = (\lambda_{\bm s_1}\lambda_{\bm s_2})\left({\lambda_{\bm s_1}^2+\lambda_{\bm s_2}^2\over2}\right)^{-1}\exp\left\{-\left({\lambda_{\bm s_1}^2+\lambda_{\bm s_2}^2\over2}\right)^{-1/2}(1+R_i)^\nu\|\bm s_1-\bm s_2\|\right\},
\end{equation}
where $\lambda_{\bm s}>0$ is a spatially-varying range parameter surface that can be associated with spatial covariates, and $\nu\geq0$ is a parameter that controls how the scaling variable $R_i$ influences the dependence range of $W_i$. {In particular, when $\nu=0$, $R_i$ and $W_i$ are independent.} In the simulation examples below, we will illustrate some specific structures for $\lambda_{\bm s}$. When the value of $R_i$ is large and $\nu>0$, the dependence within the Gaussian random field $W_i$ indeed decreases with increasing ``magnitude" $R_i$. 
In this model, the {mean} measure {of $\{\eta_i;i=1,2,\dots\}$ at the finite set of} locations in $K\equiv K_N$ may be written as
\begin{equation}
\label{eq:meanmeasure1}
{\Lambda_K([\bm 0,\bm z_K]^C)=\int_0^\infty\{1-\Phi(\bm z_K/r;\bm\rho(K,K;r))\}\kappa(\mathrm{d}r),\quad \bm z_K > \bm 0,}
\end{equation}
where $\Phi(\cdot;\bm\rho(K,K;r))$ is the multivariate normal distribution function with zero mean and correlation matrix $\bm\rho(K,K;r)$ determined by the correlation function $\rho(\bm s_1,\bm s_2;r)$ for $\bm s_1,\bm s_2\in K${, and $[\bm 0,\bm z_K]^C \equiv ([0,z_1]\times\cdots\times[0,z_N])^C$}. {\cite{Huser.etal:2020} showed that the marginal mean measure $\Lambda_{\bm s_0}$ is an infinite continuous measure and locally finite, thus it satisfies the assumptions in \eqref{eq:assumption} stated in \S\ref{sec:Max-id}.}
{
The intensity function $g_{\Lambda_K}(\bm z_K)$ of $\Lambda_K$ can be written as follows, where we use notation $\varphi$ for the multivariate Gaussian density corresponding to $\Phi$ in \eqref{eq:meanmeasure1}:
\begin{equation}\label{eq:intensity-function}
g_{\Lambda_K}(\bm z_K) = 
\int_0^\infty r^{-N}\varphi(\bm z_K/r;\bm\rho(K,K;r))\kappa(\mathrm{d}r),\quad \bm z_K > \bm 0.
\end{equation}
For marginal intensities arising for singletons $K=\{\bm s\}$, we write $g_{\Lambda_{\bm s}}$. 
The density function of the probability measure $P_{\bm s_1}(z,\cdot)$ is then given by the conditional intensity}\begin{equation}\label{eq:regular-conditional-intensity}
{\bm z_{K\setminus \{\bm s_1\}} \mapsto g_{\Lambda_K}(z, \bm z_{K\setminus \{\bm s_1\}})/g_{\bm s_1}(z) \propto g_{\Lambda_K}(z, \bm z_{K\setminus \{\bm s_1\}}).}
\end{equation}


To illustrate the tail dependence properties of this max-id process, we use the level-dependent extremal coefficient \citep{Padoan2013,Huser.etal:2020}. For a collection of sites $K_N \subset \calS$ and $|K_N|=N$, the level-dependent extremal coefficient is defined as 
\begin{equation}\label{eq:ext_coef}
\theta_N(z_0) = {\log\{G(\bm z_0)\}\over\log\{G_0(z_0)\}}\Rev{,}
\end{equation}
where $\bm z_0 = (z_0,\dots,z_0)^T\in\mathbb {R}^N$, $G_0$ is the marginal distribution of $Z$ and $G$ is the joint distribution of the random vector $\bm Z(K_N)$ with observations at the locations in the set $K\subset\calS$. This coefficient can be expressed in terms of the measure $\Lambda$, and in the case where the margins are unit Fr\'echet it is specifically equal to $\theta_N(z_0) =z_0\Lambda_{K_N}\left\{(-\infty,\infty)^N\setminus (-\infty,z_0]^N\right\}$. From \eqref{eq:ext_coef} it is easy to see that $\Pr\{\bm Z(K_N)\leq z_0\} = G_0(z_0)^{\theta_N(z_0)}$, such that this coefficient can be interpreted as the effective number of independent variables among the components of the vector $\bm Z(K_N)$ at level $z_0$.  Higher values of $\theta_N(z_0)$ imply weaker dependence at quantile level $z_0$. This dependence measure  is widely used in extreme-value studies (mainly with max-stable models), where modeling the tail dependence structure is the main interest. As an example, we consider the correlation function of $W_i$ conditional on the event magnitude $R_i$ chosen as $\rho(\bm s_1,\bm s_2;R_i) =\exp\left\{-2(1+R_i)^\nu \|\bm s_1-\bm s_2\|\right\}$, which corresponds to the stationary and isotropic counterpart of \eqref{eq:corr_function_W} with spatially-constant $\lambda_{\bm s}=1/2$.  Figure~\ref{fig:GS} shows the realizations of $\{Z(s)\}_{s\in(0,1)}$ on a standard Gumbel scale generated using our proposed exact simulation algorithm, as well as the bivariate level-dependent extremal coefficients between two sites at distance $h=0.5$ for the parameters $\alpha=1$, $\beta=0,0.5,1,2$ and $\nu=0,0.25,0.5,1$. When $\beta=0$ (interpreted here as $\beta\downarrow0$), the extremal dependence persists with increasing level $z$. The curve of extremal dependence becomes a horizontal line when we $\nu$ and $\beta$ are both zero. In this case, the max-id process is the max-stable extremal-$t$ process. When $\beta$ is different from zero, we get asymptotic independence (since $\theta_2(z_0)\to2$ when $z\to\infty$), and the dependence level decreases as the values of $\beta$ and $\nu$ increase. The realizations were generated using the same random seed, and we can discern stronger fluctuations of the realizations as the value of  $\beta$ increases, corresponding to a decreasing dependence strength. 

\begin{figure}[t!]
	\centering
	\includegraphics[width=1\linewidth]{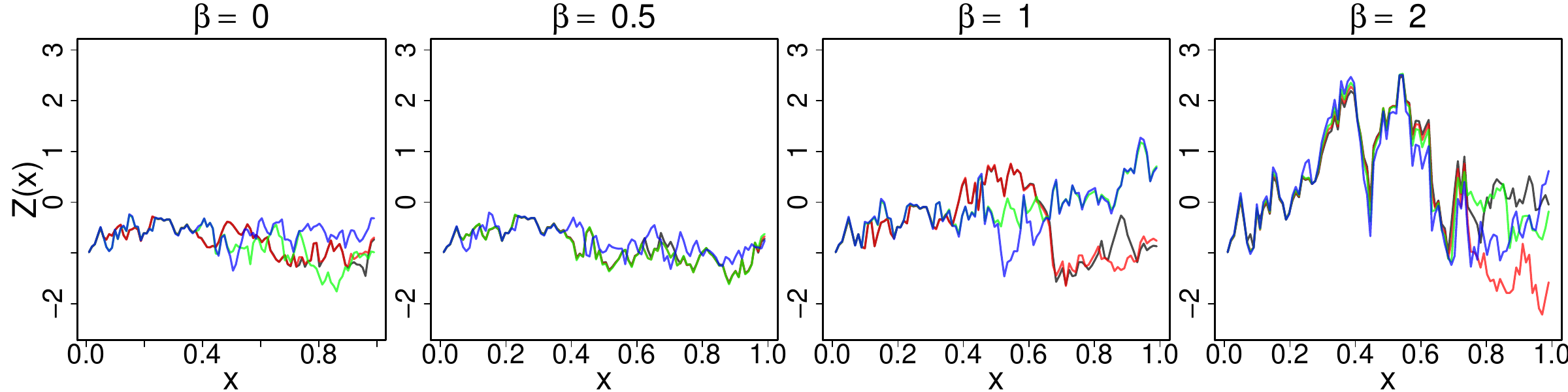}
	\includegraphics[width=1\linewidth]{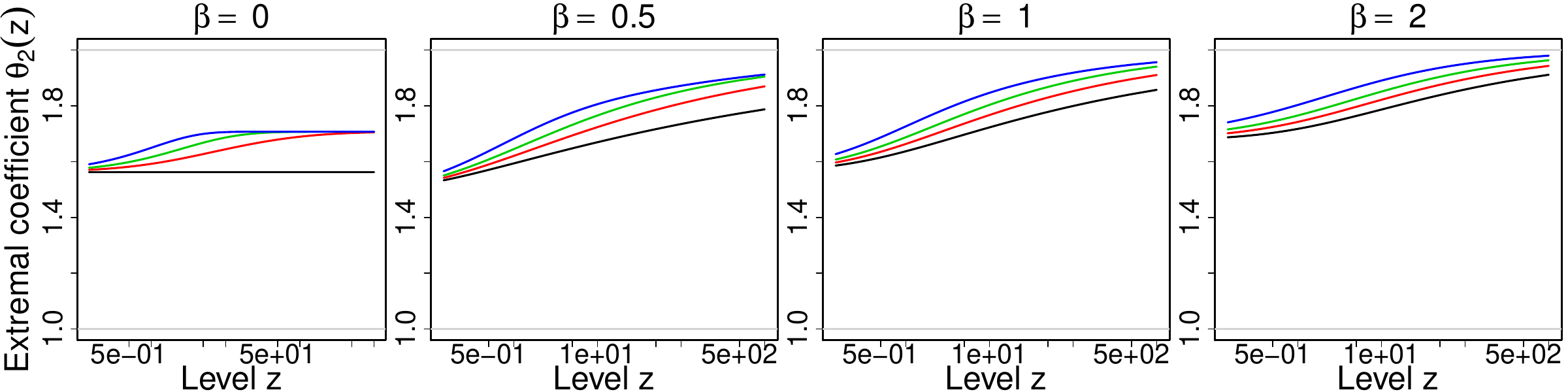}
	\caption{Simulations (top row) and bivariate level-dependent extremal coefficients (bottom row) for models based on the Gaussian scale mixture construction. The parameters are chosen as $\nu=0, 0.25,0.5,1$ (black, red, green, blue curves), where $\nu$ controls the dependence between $R_i$ and $W_i$ by using correlation function $\rho(\bm s_1,\bm s_2;R_i) =\exp\left\{-2(1+R_i)^\nu \|\bm s_1-\bm s_2\|\right\}$ conditional on $R_i$. The level $z$ is shown on a logarithmic scale.}
	\label{fig:GS}
\end{figure}

\subsubsection{Simulation algorithm}\label{sec:gaussian-scale-mixture-simulation}

The generic Algorithm~\ref{algo} based on extremal functions can be applied, and we now specifically describe how to sample from $P_{\bm s_1}(z,\cdot)$ {at a finite set of locations $K$} (need in lines 2 and 7 of Algorithm~\ref{algo}) for the max-id model constructed from Gaussian scale mixtures $\eta_i=R_iW_i$ as described in \S\ref{sec:gaussian-scale-mixture-description}. {Intuitively speaking, we first sample a random variable $\tilde R_z$ whose probability distribution corresponds to the intensity of $R_i$ conditional on $\eta_i(\bm s_1) = R_iW_i(\bm s_1)=z$, and in the second step we sample the Gaussian field $W_i$ at locations in $K$ conditional on $\eta_i(\bm s_1)=z$ and $\tilde R_z$, i.e., we sample a Gaussian random vector conditional on its value $z/ \tilde{R}_z$ at $\bm s_1$. 
} 

{To obtain the distribution of $\tilde R_z$ in the first step, we consider the joint intensity function of points $\{R_i\}$ and $\{W_i(\bm s_1)\}$ given by $(r,w)\mapsto \kappa(\mathrm{d}r)\varphi(w)$.  
 Next, we perform a change of variables from $(r,w)$ to $(r,z) =(r, rw)$, which yields the intensity $g_{R,\eta}(r,z) =\kappa(\mathrm{d}r)\varphi(z/r)/r$. Finally, we obtain the desired conditional density as $g_{\tilde R_z}(r;z) =g_{R,\eta}(r,z)/g_{\Lambda_{\bm s_1}}(z)$. 	Note that 
 this defines indeed a valid probability density function 
 since it corresponds to a conditional distribution of $P_{\bm s_1}(z,\cdot)$, which has been shown to be a proper probability measure.} 

{In the second step, conditional on $\tilde R_z$, we write $\tilde \eta(K)$ for the conditional Gaussian random vector with mean vector $\bm\rho(\bm s_1,K;\tilde R_z)z$, common variance $\tilde{R}_z^2$, and correlation matrix $\bm \rho(K,K;\tilde R_z)-\bm\rho(\bm s_1,K;\tilde R_z)\bm\rho(\bm s_1,K;\tilde R_z)^\top$, where $\bm\rho(\bm s_1,K;\tilde R_z)$ denotes the column vector of correlation values $\rho(\bm s_1,\bm s_j;\tilde R_z)$, $j=1,\ldots,N$, and $\bm \rho(K,K;\tilde R_z)$ is a matrix with entries $\rho(\bm s_i,\bm s_j;\tilde R_z)$, $i=1,\ldots,N,\, j=1,\ldots,N$. 
 	 We denote this joint Gaussian density function as $g_{\tilde \eta(K) \mid \tilde R_z}$. 
Note that, after integrating out $\tilde R_z$, it is easy to formally show that the (unconditional) distribution of $\tilde \eta(K)$ is 
$$\tilde \eta(K) \sim P_{\bm s_1}(z,\cdot).$$ 
Simulation from $P_{\bm s_1}(z,\cdot)$ at locations in $K$ can thus be performed using a two-step procedure, where we need to simulate $\tilde R_z$ first, and conditional on $\tilde R_z$, we then simulate $\tilde \eta(K)$ from a multivariate Gaussian distribution.} We here explore the use of two simulation techniques where {$g_{\tilde R_z}(r;z)$ only needs to be known up to a constant that depends on the value of $z$.} Precisely, {simulation from $P_{\bm s_1}(z,\cdot)$ proceeds} according to the following {successive} steps:
\begin{enumerate}[(1)]\label{simu:Gaussian-scale}
    \item Simulate a value {$\tilde R_z$} from {the density function $g_{\tilde R_z}(r;z)$, $r>0$}, using either 
    \begin{itemize}
    	\item the Metropolis--Hastings algorithm \citep[e.g.,][]{Robert2013} 
    	\item  the  adaptive rejection sampling method proposed in \cite{Gilks.Wild:1992}.
    \end{itemize}
    \item Given {$\tilde R_z$}, simulate a {random vector $\tilde \eta(K)$} from the {multivariate Gaussian density function $g_{\tilde \eta(K) \mid \tilde R_z}$}. 
    \item Return {$\tilde \eta(K)$}. 
\end{enumerate}
Notice that the second step corresponds to simulating  from a conditional multivariate normal distribution and can easily be performed exactly \citep[e.g.,][]{Varadhan2015} in dimension $N$ up to several thousands. 
Therefore, the key difficulty {for an exact simulation procedure} lies in the first step, where we have to simulate from the probability density {$g_{\tilde R_z}$}, {known up to a constant only.} 

\subsubsection{Options for simulating {$\tilde R_z$}}

With the Metropolis--Hastings (MH) method in Step~1 of the above algorithm in \S\ref{sec:gaussian-scale-mixture-simulation}, we use the following symmetric random walk proposal {on a logarithmic scale}: 
$${\log(\tilde R_z^{j}) \mid \tilde R_z^{j-1}  \sim \mathcal{N}(\log(\tilde R_z^{j-1}),\sigma^2), \quad j = 1,2,\ldots,}$$
where $\sigma^2$ is the proposal variance and has to be tuned {to achieve good mixing}, and the initial value {$\tilde R_z^{0}$} is either fixed by the user or drawn from some distribution that we can easily simulate from. 
MH is {usually} approximate (thus not exact) when the distribution of the initial value is different from the target distribution (which is always the case in practice), but its error becomes negligible when a sufficient number of iterations are performed. Here, the MH algorithm is performed on a single variable {$\tilde R_z$} at a time, so each Markov chain Monte Carlo (MCMC) chain {involved in Algorithm~\ref{algo}} mixes well and converges very quickly to its stationary distribution, while iterations are of very small computational cost. Our experience shows that very few iterations are in fact needed in practice to get accurate results. We therefore regard Algorithm~\ref{algo} combined with MH for simulating {$\tilde R_z$} as a general ``quasi-exact'' simulation method. {\citet{Propp.Wilson:1996} further showed how to perform exact MCMC-based sampling by coupling Markov chains, but for simplicity, we shall here only consider the standard MCMC procedure outlined above, since it is easy to implement and already able to deliver very accurate simulation samples.}  

As an alternative to MH, we also consider adaptive rejection sampling (ARS), which is an exact simulation algorithm that is often used in combination with the Gibbs sampler. It can generate exact independent samples from the target distribution, here the density {$g_{\tilde R_z}(r;z)$, $r>0$} with fixed $z$. The only requirement for using this method is that {$h(r)=\log g_{\tilde R_z}(r;z)$, $r>0$} is concave in $r$ with a connected support, and that it is continuous and differentiable everywhere within the support. In particular, for the previously presented parametric family of Gaussian scale mixture models based on \eqref{eq:measure_R}, it is straightforward to show that $h(r)$ is indeed a continuous, differentiable and concave function over its support $\mathbb R$. The ARS algorithm adaptively constructs a piecewise upper hull function $u(r)$ and a piecewise lower bounding function $l(r)$ of $h(r)$, respectively, and then uses the distribution $s(r)=\exp\{u(r)\}/\int \exp\{u(r')\} \mathrm{d}r'$ to generate proposals in the exact rejection sampling scheme. Figure \ref{fig:ars} illustrates the adaptive rejection sampling scheme, where the upper hull function $u(r)$ is tangent to the target $h(r)$ at {the points} $\bm r=\{r_1,r_2,r_3, r_4\}$. 
\begin{figure}
	\centering 
	\includegraphics[width=0.5\textwidth]{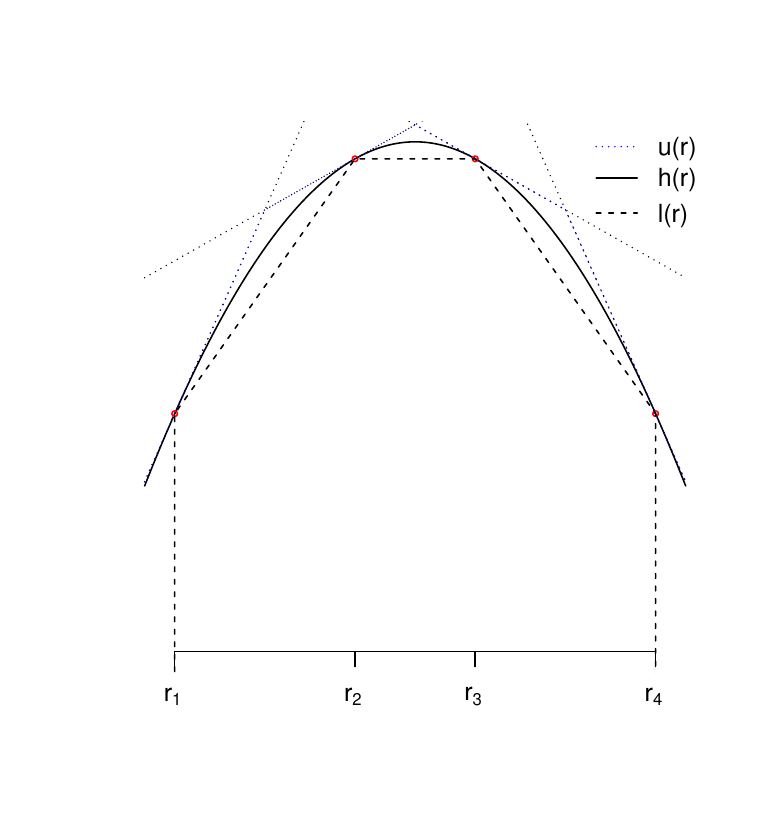}
	\caption{An illustration of the adaptive rejection sampling scheme, where the upper hull function $u(r)$ is tangent to the logarithmic target density $h(r)$ at locations $\bm r=\{r_1,r_2,r_3, r_4\}$. }
	\label{fig:ars}
\end{figure}
Therefore, $u(r)$ and $l(r)$ are uniquely determined by the {points} in $\bm r$ and the target function {$g_{\tilde R_z}(r;z)$}. As we increase the number of points in $\bm r$, the functions $u(r)$ and $l(r)$ will form an increasingly narrow envelope around the function $h(r)$, and the acceptance rate of the algorithm will increase. The steps of the ARS algorithm are as follows:
\begin{enumerate}[(1)]
   \item[(0)] Initialization: choose a set of anchor points $\bm r$ and construct the upper and lower bounds $u(r)$ and $l(r)$, respectively, as illustrated in Figure~\ref{fig:ars}.
   \item Sampling: sample {$\tilde R^*_z \sim s(r)$}, and sample  $U \sim \rm{Unif}(0,1)$ independently of {$\tilde R^*_z$}. 
   \item First acceptance test:  if {$U \leq \exp\{l(\tilde R^*_z)-u(\tilde R^*_z)\}$} then accept {$\tilde R_z^*$} and go to 5. 
   \item Second acceptance test: If {$U \leq \exp\{h(\tilde R^*_z)-u(\tilde R^*_z)\}$} then accept {$\tilde R_z^*$} and go to 5.
   \item If the first and second acceptance tests both fail, reject {$\tilde R^*_z$}; add {$\tilde R^*_z$} to list of anchor points and update the bounds $u(r)$ and $l(r)$. Go to 1.
   \item Return {$\tilde R^*_z$}.
\end{enumerate}
The distribution of the simulated value given acceptance is thus 
\begin{eqnarray*}
\Pr(\tilde R^*_z \in A\mid\text{acceptance of }\tilde R^*_z) & = &{\int_A s(r) \left\{\int_{0}^{\exp\{l(r)-u(r)\}}{\rm d}w + \int_{\exp\{l(r)-u(r)\}}^{\exp\{h(r)-u(r)\}}{\rm d}w\right\}{\rm d}r\over\int_{\mathbb R} s(r) \left\{\int_{0}^{\exp\{l(r)-u(r)\}}{\rm d}w + \int_{\exp\{l(r)-u(r)\}}^{\exp\{h(r)-u(r)\}}{\rm d}w\right\}{\rm d}r}\\
&=&{\int_A \exp\{h(r)\}{\rm d}r\over\int_{\mathbb R} \exp\{h(r)\}{\rm d}r},
\end{eqnarray*}
which corresponds  exactly to  the target distribution with density {$g_{\tilde R_z}(r;z),r>0$}. By exploiting the concavity property of $h$, ARS usually only needs a relatively small number of iterations before acceptance, and it may therefore be faster and more accurate than MH, as we will illustrate through the simulation experiments in \S\ref{sec:Simu-Study}. 

\subsection{Representation of $\eta_i$ as a Gaussian location mixture} 

\subsubsection{Model description}

Another flexible and general max-id model family is the class of \emph{Gaussian location mixtures}, where $\eta_i$ has the form $\eta_i=R_i+W_i$. Here, the location variables $\{R_i\}$ are the points of a Poisson point process over $\mathbb R$ with mean measure $\overline \kappa$, and $\{W_i\}$ are independent {standard} Gaussian random fields {with mean zero and unit variance}. Again, the random variables $\{R_i\}$ can be interpreted as overall event magnitudes, while the random processes $\{W_i\}$ vary over $\mathcal{S}$ and can be viewed as spatial ``profiles".  For a general {and flexible} parametric form of the distribution of event magnitudes $\{R_i\}$,  we propose to define $\overline \kappa$ as
\begin{multline}\label{eq:measure_R2}
\overline\kappa(r) = \exp\left[-\alpha\left\{\mathbb{I}(r>0)|r|^{\beta_1}-\mathbb{I}(r<0)|r|^{\beta_2}\right\}\right], \ r\in\mathbb R,\ (\alpha,\beta_1,\beta_2)^\top\in(0,\infty)^2\times(0,2),
\end{multline}
where $\overline\kappa(r)$ is the short-hand notation for $\overline\kappa([r,\infty))$. The parameter restrictions for $\alpha$, $\beta_1$ and $\beta_2$ ensure that the mean measure $\Lambda$ of the points $\{\eta_i\}$ is indeed Radon (i.e., locally finite) as required to defined a valid Poisson point process. Max-id models based on Gaussian location mixtures have never been explored in the literature, and the new model~\eqref{eq:measure_R2} is especially interesting in view of the flexible dependence structure that it yields and its close connection to {the max-stable Brown--Resnick process} \citep{Kabluchko.etal:2009}. The parameter $\beta_1$ controls the dependence strength in the upper tail, while $\beta_2$ controls the dependence strength in the lower tail.  As in \S\ref{sec:gaussian-scale-mixture-description}, we may further allow for dependence between the point $R_i$ and the Gaussian field $W_i$ for each $i$. For a general and flexible correlation function of $W_i$, we may again use the correlation function in \eqref{eq:corr_function_W} to drive the range of the Gaussian fields $W_i$ through the location variable $R_i$, though other options are of course possible. In this case, the mean measure {of $\{\eta_i;i=1,2,\dots\}$ at the finite set of locations $K\equiv K_N$} is given by 
$${\Lambda_K([0,\bm z_K]^C)=\int_{\mathbb R}\{1-\Phi(\bm z_K-r;\bm \rho(K,K;r))\}\overline\kappa({\rm d}r),}$$
{where $\Phi(\cdot;\bm \rho(K,K;r)))$ and $\bm \rho(K,K;r)$ are defined as before. The expressions for unconditional and conditional intensity functions can be derived by analogy with the scale mixture case.}
When $\alpha=\beta_1=\beta_2=1$ and $\{R_i\}$ and $\{W_i\}$ are independent of each other, the proposed model reduces to the popular Brown--Resnick max-stable model \citep{Kabluchko.etal:2009,Wadsworth2014}, although this model is often defined through intrinsically stationary Gaussian processes $\{W_i\}$ with some variogram function $\gamma(h)$ rather than through a correlation function $\rho(\bm s_1,\bm s_2)$. 

\subsubsection{Simulation algorithm}

By analogy with the approach developed for Gaussian scale mixtures, we can here again use Algorithm~\ref{algo} for the max-id location mixture process. Moreover, to sample from {$P_{\bm s_1}(z,\cdot)$ at a finite set of locations}, we propose a similar simulation algorithm, where the first step consists in sampling a random variable $\tilde R_z$ from the probability density {$g_{\tilde R_z}(r;z)\propto \varphi(z-r)\mathrm{d}\overline\kappa({\rm d}r)$, $r>0$, and the second step is to sample a Gaussian random vector $\tilde \eta(K)$, conditional on $\tilde R_z$, with mean $\tilde R_z+\bm\rho(\bm s_1,K;\tilde R_z)(z-\tilde R_z)$, common unit variance, and correlation matrix $\bm \rho(K,K;\tilde R_z)-\bm\rho(\bm s_1,K;\tilde R_z)\bm\rho(\bm s_1,K;\tilde R_z)^\top$, where $\bm\rho(\bm s_1,K;\tilde R_z)$ denotes the column vector of correlation values $\rho(\bm s_1,\bm s_j;\tilde R_z)$, $j=1,\ldots,N$, and $\bm \rho(K,K;\tilde R_z)$ is the corresponding correlation matrix. As above, we denote the density of $\tilde \eta(K)$ conditional on $\tilde R_z$ as $g_{\tilde \eta(K) \mid \tilde R_z}$. Again, the resulting (unconditional) distribution of $\tilde \eta(K)$ can be shown to be $P_{\bm s_1}(z,\cdot)$ at locations in $K$.} However, here we observe that the density {$g_{\tilde R_z}(r;z)$} is not logarithmically concave, and therefore we cannot use the adaptive resampling scheme to generate extremal functions {at a finite set of locations in $K$} from {$P_{\bm s_1}(z,\cdot)$}, but we can still use the quasi-exact Metropolis--Hastings algorithm, which provides high accuracy at low computational cost. Thus, our proposed procedure for simulating max-id processes of Gaussian location mixture type at a collection of sites {$K$} relies on Algorithm~\ref{algo} combined with the following {successive steps to sample from $P_{\bm s_1}(z,\cdot)$}:
\begin{enumerate}[(1)]\label{simu:Gaussian-location}
    \item Simulate a value  {$\tilde R_z$} according to the {density $g_{\tilde R_z}(r;z)$} by using the Metropolis--Hastings method with random walk proposal distribution {$\tilde R_z^{j}\mid \tilde R_z^{j-1} \sim \mathcal{N}(\tilde R_z^{j-1},\sigma^2)$, $j=1,2,\ldots$}, and initial point $\tilde R_z^0$.
    \item Given {$\tilde R_z$}, simulate a vector {$\tilde \eta(K)$} according to the {multivariate Gaussian density function $g_{\tilde \eta(K) \mid R_z}$}.
    \item Return {$\tilde \eta(K)$}. 
\end{enumerate}
As in the Gaussian scale mixture case, the second step consists in simulating from a conditional multivariate normal distribution, which is fast and straightforward. Figure~\ref{fig:GL} shows examples of realizations of $\{Z(s)\}_{s\in(0,1)}$ with standard Gumbel marginal distributions (generated using the proposed quasi-exact algorithm), as well as the bivariate level-dependent  extremal coefficients at distance $h=0.5$ for the parameters $\beta_1 = 0.5,1,1.5$, $\beta_2 = 0.5,1,1.5$ and $\alpha=1$. The correlation function for $W_i$ is here chosen to be exponential as $\rho(h) =\exp\left(-2h\right)$. We see that when both $\beta_1$ and $\beta_2$ are equal $1$, the curve of the extremal coefficients is a horizontal line. In this case, the max-id process is the Brown--Resnick max-stable process. As the values for $\beta_1$ and $\beta_2$ increase, the dependence levels in the lower tail and the upper tail, which are controlled by $\beta_2$ and $\beta_1$, respectively, decrease. When $\beta_1>1$, the extremal coefficient increases as the level $z\to\infty$. By contrast, when $\beta_1<1$, the extremal coefficient decreases as the level $z\to\infty$. Similar results are observed for the dependence level in the lower tail with $\beta_2$. Therefore, this model shows great flexibility through its ability to capture asymptotic dependence and asymptotic independence in the upper tail and different extremal dependence structures in the lower tail. 
In these  realizations generated with the same random seed, we discern that the fluctuations are increasing as the dependence level decreases, as is expected.

\begin{figure}[t!]
	\centering
	\includegraphics[width=1\linewidth]{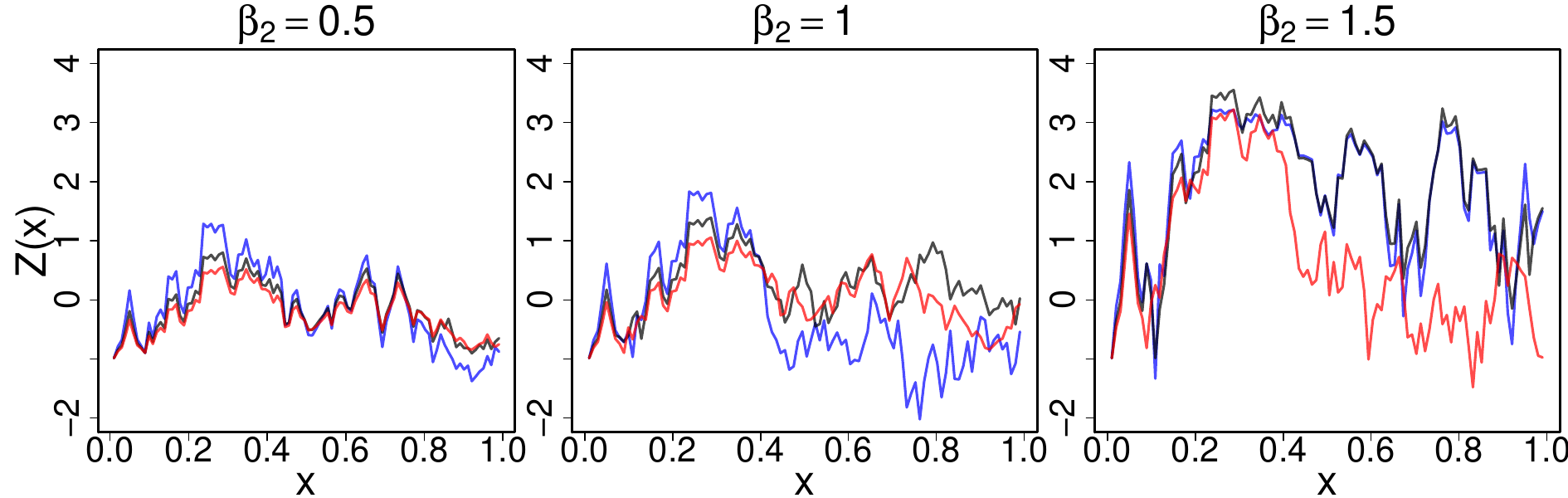}
	\includegraphics[width=1\linewidth]{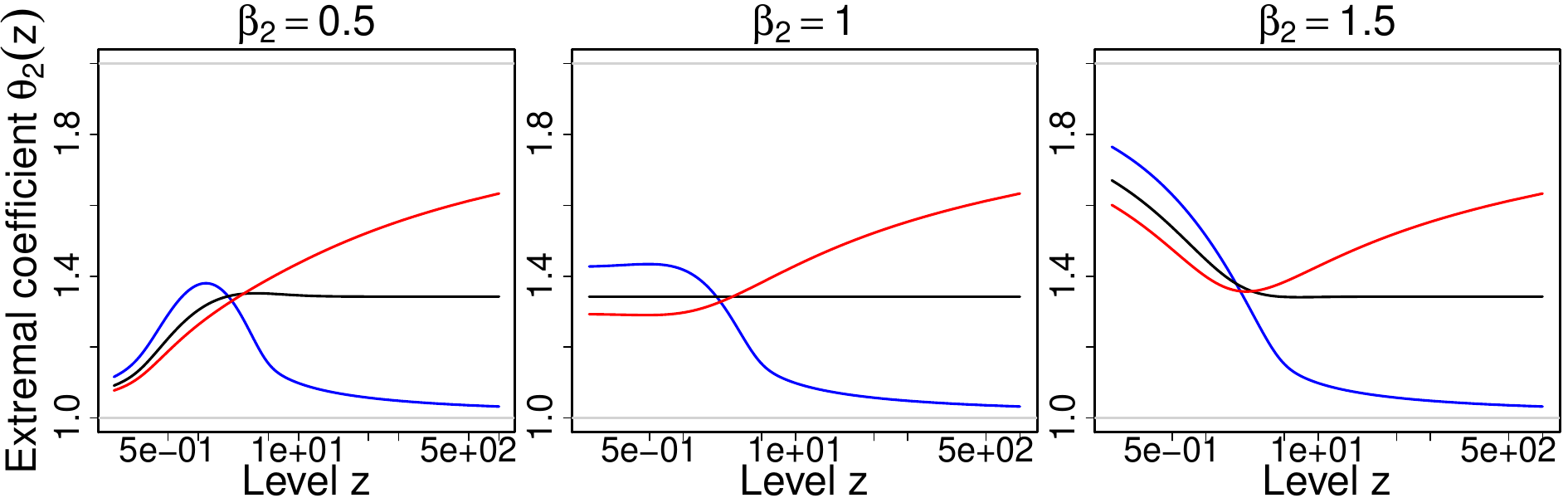}
	\caption{Simulations (top row) and bivariate level-dependent extremal coefficients (bottom row) for models based on the Gaussian location mixture construction. Parameters are chosen as $\alpha=1$ and $\beta_1=0.5,1,1.5$ (blue, black, red). The correlation function for $W_i$ is exponential and given by $\rho(h) =\exp\left(-2h\right)$. The level $z$ is shown in logarithmic scale.}
	\label{fig:GL}
\end{figure}
	
\section{Simulation Study}
\label{sec:Simu-Study}

In this section, we use the Gaussian scale mixture representation described in \S\ref{sec:gaussian-scale-mixture} as an example to illustrate Algorithm~\ref{algo} and to assess the performance of the two algorithmic choices detailed  in \S\ref{sec:gaussian-scale-mixture-simulation} for sampling from $P_{\bm s_0}(z,\cdot)$, i.e., of the Metropolis--Hastings scheme and the adaptive rejection sampling scheme. We also provide a comparison with respect to a naive simulation approach inspired from \citet{Schlather:2002}, which consists in simulating the magnitudes $R_i$ in descending order and truncating the maximum in \eqref{eq:max-id} as $Z(\bm s)\approx\max_{i=1,\ldots,n}\eta_i(\bm s)$ with $\eta_i=R_iW_i(\bm s)$ for a sufficiently large integer $n$. 

 We perform the simulation on a $7\times7$ regular grid in $(0,1)^2$, corresponding to $49$ locations in $K$. We fix $\alpha=5,\beta=2$ in \eqref{eq:measure_R} and $\nu=3,\lambda_{\bm s} = \exp[2-0.5\Phi\{(s_x-0.5)/0.25\}]$ in \eqref{eq:corr_function_W}, where $\bm s=(s_x,s_y)^\top$ and $\Phi(\cdot)$ is the standard normal distribution function. We implemented three algorithms using the \texttt{R} software: the approximate naive sampling method (APPROX), Algorithm~\ref{algo} combined with the quasi-exact random walk Metropolis--Hastings (MH), and Algorithm~\ref{algo} combined with the exact adaptive rejection sampling (ARS). Our aim is to compare the simulation accuracy with respect to univariate distributions and to dependence properties. 
 We tuned the settings of these three algorithms such that they require approximately the same amount of running time to generate $N=10^5$ independent samples of the max-id process on the grid. These settings are given as follows: 
\begin{enumerate}[(1)]
	\item Naive method: Simulate $\eta_i=R_iW_i$,  $i=1,2,\dots,n$ on the grid  $K$, where $n=100$.
	\item Metropolis--Hastings scheme: In the MCMC, run $100$ iterations to generate one sample {$\tilde R_z$} from {the density function $g_{\tilde R_z}(r;z)$}.
	\item Adaptive rejection sampling scheme: Use 6 points to initialize the upper hull function and lower hull function.
\end{enumerate}
Based on  the generated samples, we compute the empirical Kullback--Leibler (KL) divergence $D_{\bm s_0}(p\|q)$ \citep{Kullback.Leibler:1951} in order to compare the accuracy of simulated samples at each site $\bm s_0\in K\subset\calS$. Precisely, we compute
\begin{equation}
	{D_{\bm s_0}(p\|q) = \sum_{i=1}^N \hat p\{Z_i(\bm s_0)\}\log\left[{\hat p\{Z_i(\bm s_0)\}\over q\{Z_i(\bm s_0)\}}\right],\quad \hat p\{Z_i(\bm s_0)\} = {p\{Z_i(\bm s_0)\}\over\sum_{i=1}^N p\{(Z_i(\bm s_0))\}},}
\end{equation}
{where $q$ places mass $1/N$ on each simulated sample as the empirical cdf,} while $p$ is the exact marginal density of $Z$. {In practice, the samples are always treated as discrete random variables, and $\hat p$ and $q$ need to be probability mass functions instead of density functions.} Therefore, we have $49$ empirical KL divergences in total, one for each point on the grid. Figure \ref{fig:KL} shows the boxplot of the $49$ empirical KL divergences using the simulated datasets generated by the three algorithms. Both the adaptive rejection sampling scheme and the Metropolis--Hastings scheme outperform the naive method by a significant margin. Moreover, the median of the empirical KL divergence for the adaptive rejection sampling scheme is notably lower than  with  the Metropolis--Hastings scheme. Generating a single max-id replication on the grid took approximately 1.2 seconds using a single-core machine. 
\begin{figure}[t!]
	\centering
	\includegraphics[width=0.5\textwidth]{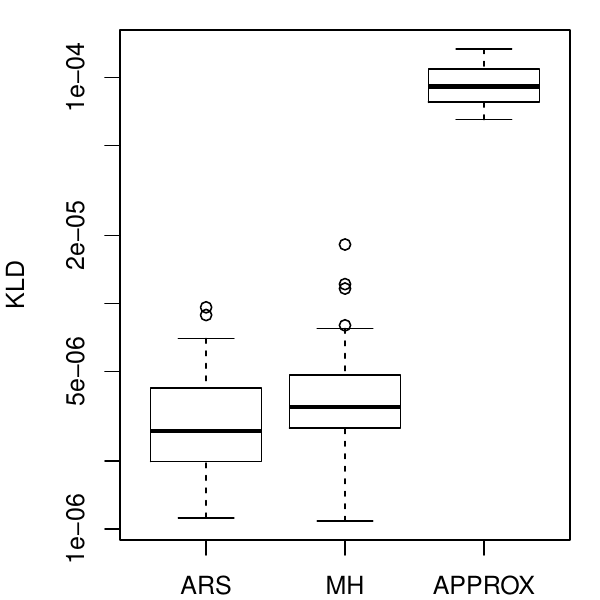}
	\caption{Boxplots of  empirical KL divergences for the $49$ grid locations based on the simulated data sets generated by the naive method (APPROX), the Metropolis--Hastings scheme (MH) and the adaptive rejection sampling scheme (ARS).}\label{fig:KL}
\end{figure}

To further evaluate the performance of each sampling scheme with respect to the dependence structure, we calculated the bivariate (level-dependent) empirical extremal coefficients $\hat\theta_2$. For each pair of sites, we compared the empirical extremal coefficient computed from the $N=10^5$ simulated replicates with the theoretical counterpart computed by evaluating the mean measure $\Lambda$ with a very accurate numerical unidimensional integral.  Figure~\ref{fig:ext_coef} shows the results for the three simulation techniques {with the computed standard error, $sd$, of the approximation error shown in the title of each panel. The adaptive rejection sampling scheme appears to perform generally better than the naive scheme at all levels, especially at low quantile levels when the naive scheme is significantly biased. At high quantile levels ($q=0.95$), all methods have a comparable performance. The pseudo-exact Metropolis--Hastings scheme appears to be as accurate as the adaptive rejection sampling method.} Our results imply that if we draw statistical inferences based on simulated samples from the naive method, the strength of the dependence at relatively low quantile levels of extremes would be significantly overestimated in the current setting. Better results might be obtained for the naive method by increasing the number $n$ of simulated fields $\{\eta_i\}$, but this comes with a computational cost. By contrast, our proposed algorithm always provides exact max-id simulations at a low computational cost, which may prove essential for simulation-based statistical inference.

\begin{figure}
	\centering
	\includegraphics[page=1,width=0.9\textwidth]{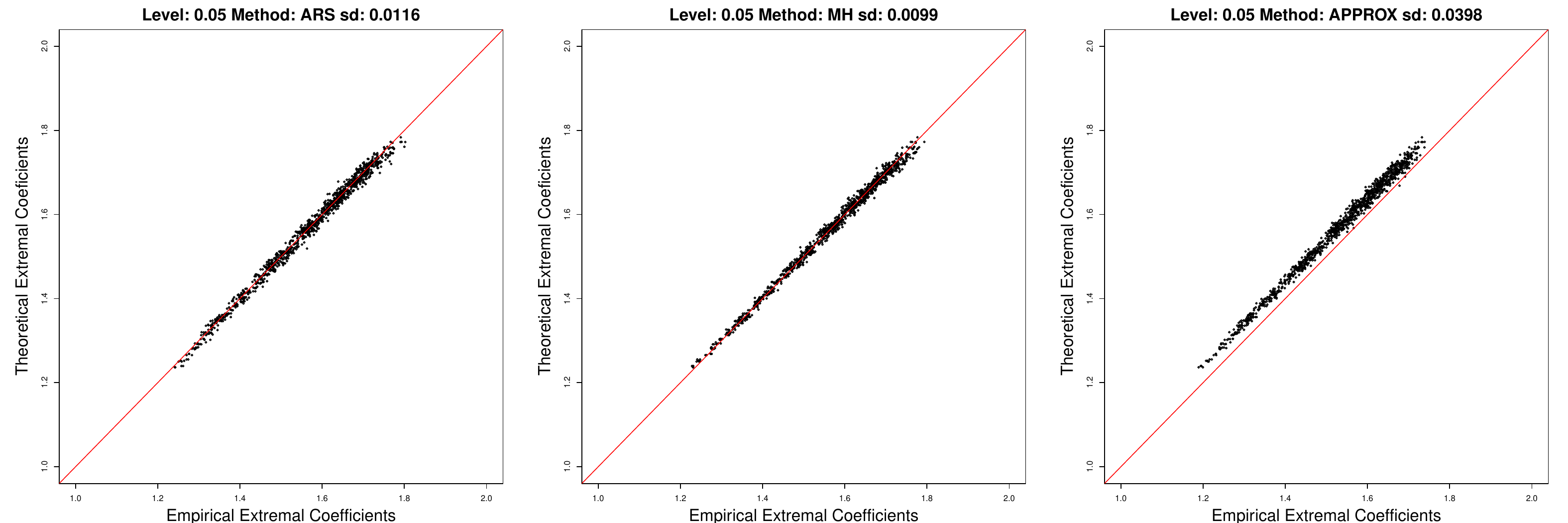}
	\includegraphics[page=2,width=0.9\textwidth]{scatter_plot_extremal.pdf}
	\includegraphics[page=4,width=0.9\textwidth]{scatter_plot_extremal.pdf}
	\includegraphics[page=5,width=0.9\textwidth]{scatter_plot_extremal.pdf}
	\caption{Bivariate scatterplots of theoretical against empirical bivariate level-dependent extremal coefficients at quantile levels $0.05,0.25,0.5,0.75,0.95$ (from top to bottom) for all pairs of locations, based on the three sampling schemes (left: Algorithm~\ref{algo} with adaptive rejection sampling (ARS), middle: Algorithm~\ref{algo} with Metropolis--Hastings (MH), right: naive approximate method).}\label{fig:ext_coef}
\end{figure} 

As in \citet{Dombry.Engelke.Oesting:2016}, it can be shown that the expected number of random fields $W_i$ to be simulated in the exact simulation algorithm equals to the number of simulation sites, i.e., the dimension $|K|$. To further illustrate the efficiency of our proposed algorithm, we simulated 10 replicas for different numbers of randomly sampled locations $K$ on the unit square $(0,1)^2$ using the ARS scheme and calculated the time and number of random fields $W_i$ that needed to be simulated. The results are shown in Figure~\ref{fig:complexity}. The parameters are set to be $\alpha=1,\beta=1$ and the correlation function for $W_i$ is chosen to be $\rho(h) = \exp(-h)$. Since we use the Cholesky decomposition for simulating the process $W_i$ on $K$, whose complexity is $\calO(|K|^3)$, the expected overall complexity for simulating a max-id replicate is $\calO(|K|^4)$. The red curve on the left panel of Figure~\ref{fig:complexity} illustrates this. Here, the computational time is about two minutes for one replicate in dimension $|K|=900$, which is already quite fast. If further speeds-up are required, approximate algorithms for simulating the Gaussian processes $W_i$ may be employed (e.g., based on Gaussian Markov random fields). As expected, the right panel of Figure~\ref{fig:complexity} shows that the expected number of simulations of $W_i$ that are needed is indeed equal to $|K|$. We also investigated how the ordering of locations $\bm s_1,\ldots,\bm s_n$ chosen in Algorithm~\ref{algo} to iteratively sample from $P_{\bm s_0}(z,\cdot)$ might affect the algorithm's efficiency. Our results (not shown) suggest that there was very little difference among the four orderings considered (coordinate-wise, random, middle-out, or maximum-minimum orderings), in terms of the mean and variance of the number of processes $W_i$ to sample in each simulation. {In all simulations presented here,} we used a coordinate-wise ordering for simplicity.

\begin{figure}[t!]
	\centering
	\includegraphics[width=0.8\textwidth]{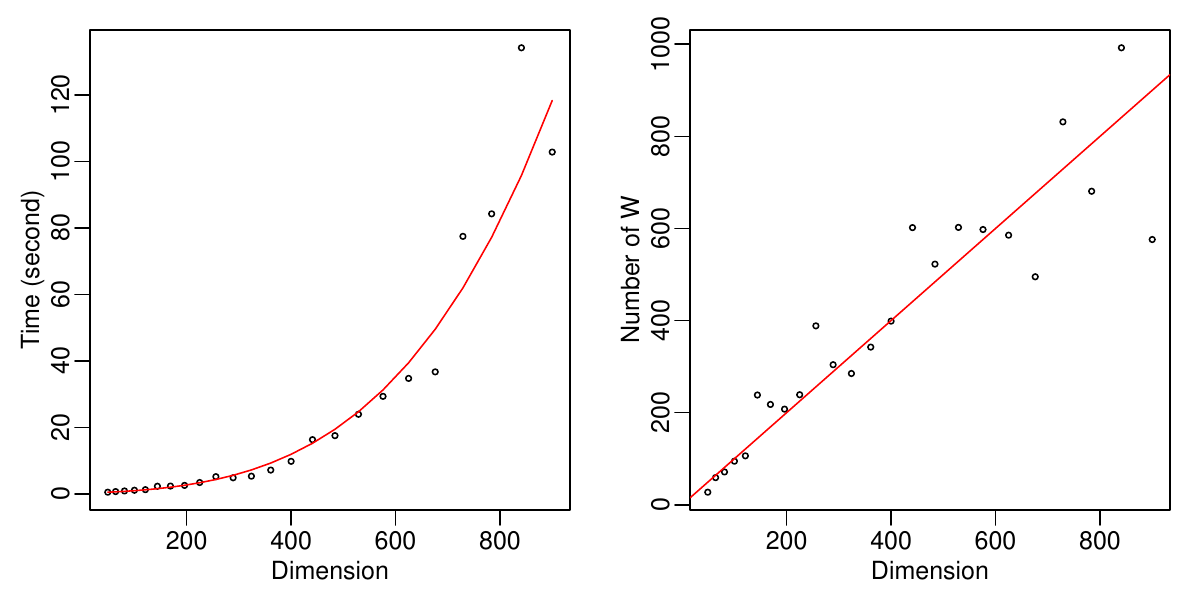}
	\caption{Left: Time in second needed to simulate one max-id replicate for different numbers of locations $|K|$. The red curve is the fit from a $4$-th order polynomial. Right: Number of random fields $W_i$ that were simulated for one max-id replicate for different numbers of locations $|K|$. The red line is the diagonal. The parameters are here set to be $\alpha=1,\beta=1$ and the correlation function for $W_i$ is chosen to be $\rho(h) = \exp(-h)$.}
	\label{fig:complexity}
\end{figure}
\section{Discussion}
\label{sec:Discussion}

In this paper, we have generalized the exact simulation algorithm for max-stable processes based on extremal functions to the more general class of max-id processes. Whenever possible, we recommend using the adaptive rejection sampling scheme for the conditional simulation of {$\tilde R_z$}, in order to achieve fast and exact simulations in cases where variables to be simulated do not correspond to some standard distribution from which samples are easily drawn. Although the specific implementation of the algorithm still depends on the chosen representation for $\eta_i$, our approach is general and can be applied to very broad classes of max-id processes while preserving computational efficiency. Moreover, the specific examples showcased in this paper already have very flexible dependence structures. In a simulation study, we used three different sampling algorithms to compare their speed and accuracy. Our novel exact simulation algorithm using the adaptive rejection sampling scheme clearly outperforms the two alternatives. In cases where the target density in the adaptive rejection sampling scheme is not logarithmically concave, it is possible to use an appropriately defined quasi-exact Metropolis--Hastings scheme to achieve high accuracy. {Exploring exact MCMC procedures based on coupling of Markov chains \citep{Propp.Wilson:1996} would also be an interesting future research direction.}

Besides standard uses of such fast and exact simulation techniques in Monte--Carlo-based statistical inferences involving quantities that cannot be calculated analytically, another interesting perspective concerns Approximate Bayesian Computation (ABC) for general max-id processes, which could be implemented by analogy with the approach developed for max-stable processes in \citet{Erhardt2012}. {As described above, the computation of multivariate density and distribution functions} is often numerically very expensive due to integrals that have to be evaluated numerically, such that likelihood-free inference techniques are promising alternatives. Fast simulation of large numbers of samples of the model is required in likelihood-free estimation methods such as ABC, and in other simulation-based inference methods such as maximum simulated likelihood \citep[e.g.,][]{Hajivassiliou2000}.

\section*{Acknowledgments}
{We thank the Editor, Associate Editor, and two anonymous reviewers for constructive comments that helped improve the paper. This publication is based upon work supported by the King Abdullah University of Science and Technology (KAUST) Office of Sponsored Research (OSR) under Awards No. OSR-CRG2017-3434 and No. OSR-CRG2020-4394.}

\baselineskip=15pt

\bibliographystyle{CUP}
\bibliography{reference}
\baselineskip 10pt

\end{document}

%% file: macros.tex
\def\frac#1#2{{\textstyle{#1\over#2}}}

\DeclareSymbolFont{AMSb}{U}{msb}{m}{n}
\DeclareMathSymbol{\Natural}{\mathbin}{AMSb}{"4E}
\DeclareMathSymbol{\Integer}{\mathbin}{AMSb}{"5A}
\DeclareMathSymbol{\Real}{\mathbin}{AMSb}{"52}
\DeclareMathSymbol{\Rational}{\mathbin}{AMSb}{"51}
\DeclareMathSymbol{\Imaginary}{\mathbin}{AMSb}{"49}
\DeclareMathSymbol{\Complex}{\mathbin}{AMSb}{"43} 
\DeclareMathSymbol{\Disk}{\mathbin}{AMSb}{"44} 
\def\bi{\begin{itemize}}
\def\ei{\end{itemize}}
\def\bd{\begin{description}}
\def\ed{\end{description}}
\def\ben{\begin{enumerate}}
\def\een{\end{enumerate}}
\def\bv{\begin{verbatim}}
\def\ev{\end{verbatim}}

\def\calC{{\mathcal C}}

\def\calO{{\mathcal{O}}}
\def\calS{{\mathcal{S}}}

\def\hat#1{{\widehat{#1}}}

\def\pr{{\rm Pr}}
\def\Pr{\pr}
\def\E{{\rm E}}

\def\2to{{\ {\buildrel 2\over \longrightarrow}\ }}

\def\I1ton{{$I_1,\ldots,I_n$}}
\def\X1ton{{$X_1,\ldots,X_n$}}
\def\Y1ton{{$Y_1,\ldots,Y_n$}}
\def\Z1ton{{$Z_1,\ldots,Z_n$}}
\def\R1ton{{$R_1,\ldots,R_n$}}
\def\e1ton{{$e_1,\ldots,e_n$}}
\def\t1ton{{$t_1,\ldots,t_n$}}
\def\x1ton{{$x_1,\ldots,x_n$}}
\def\y1ton{{$y_1,\ldots,y_n$}}
\def\z1ton{{$z_1,\ldots,z_n$}}

\def\calS{{\mathcal{S}}}

%% file: teachingmacros.tex
\newtheorem{defn}{Definition}

\newtheorem{theorem}[defn]{Theorem}
\newtheorem{lemma}[defn]{Lemma}